%% file: problem_formulation_version2.tex
\documentclass[conference,10 pt]{IEEEtran}

\usepackage{amssymb}
\usepackage{amsmath}
\usepackage{cite}
\usepackage{url}
\usepackage{xcolor}
\usepackage{subfigure}
\usepackage{cite,graphicx,amsmath,amssymb}
\usepackage{subfigure}
\usepackage{fancyhdr}
\usepackage{float}
\usepackage{mdwmath}
\usepackage{mdwtab}
\usepackage{caption}
\usepackage{amsthm}
\usepackage{setspace}
\usepackage{algorithm}
\usepackage{algorithmic}
\usepackage{makecell}
\usepackage{diagbox}
\usepackage{stfloats}
\usepackage{float}
\usepackage{picinpar}
\usepackage[colorlinks,
linkcolor=black,
anchorcolor=black,
citecolor=black
]{hyperref}
\allowdisplaybreaks
\makeatletter
\def\ScaleIfNeeded{%
	\ifdim\Gin@nat@width>\linewidth \linewidth \else \Gin@nat@width
	\fi } \makeatother
\usepackage{booktabs}
\usepackage{multirow}
\usepackage{graphicx}%??
\usepackage{epsfig}
\usepackage{stfloats}%??????

\allowdisplaybreaks[1]
\allowdisplaybreaks[4]
\usepackage{amsmath}

\title{Semantic-Relay-Aided Text Transmission: Placement Optimization and Bandwidth Allocation}

\author{\IEEEauthorblockN{Tianyu Liu\IEEEauthorrefmark{1}, Changsheng You\IEEEauthorrefmark{1}, Zeyang Hu\IEEEauthorrefmark{1}, Chenyu Wu\IEEEauthorrefmark{2}, Yi Gong\IEEEauthorrefmark{1}, and Kaibin Huang\IEEEauthorrefmark{3}} % 1\textsuperscript{st}
\vspace{0.2cm}
\IEEEauthorblockA{ \IEEEauthorrefmark{1}\text{Department of Electronic and Electrical Engineering}, \text{Southern University of Science and Technology}, Shenzhen, China\\
}
\IEEEauthorblockA{\IEEEauthorrefmark{2}\text{School of Electronic and Information Engineering}, \text{Harbin Institute of Technology}, Harbin, China\\
}
\IEEEauthorblockA{\IEEEauthorrefmark{3}\text{Department of Electrical and Electronic Engineering}, \text{The University of Hong Kong}, Hong Kong\\
}
Emails: \{liuty2022, huzy2022\}@mail.sustech.edu.cn, \{youcs, gongy\}@sustech.edu.cn, \\wuchenyu@hit.edu.cn, huangkb@eee.hku.hk.
}
\usepackage[
top    =  0.73  in,
bottom = 1.03 in,
left   = 0.57 in,
right  = 0.572 in]{geometry}

\begin{document}
\include{header}

\maketitle

\begin{abstract}
Semantic communication has emerged as a promising technology to break the Shannon limit by extracting the meaning of source data and sending relevant semantic information only. However, some mobile devices may have limited computation and storage resources, which renders it difficult to deploy and implement the resource-demanding deep learning based semantic encoder/decoder. To tackle this challenge, we propose in this paper a new \emph{semantic relay} (SemRelay), which is equipped with a semantic receiver for assisting text transmission from a resource-abundant base station (BS) to a resource-constrained mobile device. Specifically, the SemRelay first decodes the semantic information sent by the BS (with a semantic transmitter) and then forwards it to the user by adopting conventional bit transmission, hence effectively improving the text transmission efficiency. We formulate an optimization problem to maximize the achievable (effective) bit rate by jointly designing the SemRelay placement and bandwidth allocation. Although this problem is non-convex and generally difficult to solve, we propose an efficient penalty-based algorithm to obtain a high-quality suboptimal solution. Numerical results show the close-to-optimal performance of the proposed algorithm as well as significant rate performance gain of the proposed SemRelay over conventional decode-and-forward relay.
\end{abstract}


\vspace{-0.3cm}
\section{Introduction}

Semantic communication has emerged as a promising technology to improve transmission efficiency in future sixth-generation (6G) wireless systems. Specifically, by delivering semantic meaning contained in the source rather than reliably transmitting bit sequence, semantic communication significantly reduces the communication overhead and improves the resource utilization efficiency \cite{P_Zhang,W_Tong}.

In the existing literature, various approaches (e.g., deep learning (DL), knowledge graph) have been proposed to improve the semantic communication performance \cite{Z_Yang, Huiqiang_Xie,z_Weng,8,9}. Specifically, the authors in \cite{Huiqiang_Xie} and \cite{z_Weng} proposed a DL-based end-to-end semantic communication system (DeepSC) for text and speech transmissions, respectively, by jointly designing semantic-and-channel coding methods. It was shown that semantic communication can achieve superior transmission performance over conventional bit transmission, especially in low signal-to-noise ratio (SNR) and small-bandwidth regions. These works were further extended in \cite{9} by devising multi-modal semantic communication systems, referred to as U-DeepSC. To characterize the semantic communication efficiency, a new performance metric, called semantic rate, was proposed in \cite{Lei_Yan}, which measures the amount of semantic information effectively transmitted per second. Besides system architecture designs, some initial efforts have also been devoted to designing resource management for optimizing the performance of semantic communication systems. For example, the authors in \cite{Xidong_Mu} studied the optimal resource allocation policy for a heterogeneous semantic and bit transmission system, and characterized the boundary of the semantic-versus-bit rate region achieved by different multiple access schemes. In addition, the authors in \cite{10} proposed a quality-of-experience (QoE) aware resource allocation scheme to maximize the QoE by jointly designing the number of transmitted semantic symbols, channel assignment, and power allocation. However, the existing works have mostly assumed DL neural networks (e.g., DeepSC receiver) deployed and executed at the mobile devices. This, however, overlooks a key fact that some mobile devices may have limited computation and storage resources, which are incapable of performing DL-based semantic communications.

\begin{figure}[!t]
	\centering
	\includegraphics[height=4.5cm, width=8.5cm]{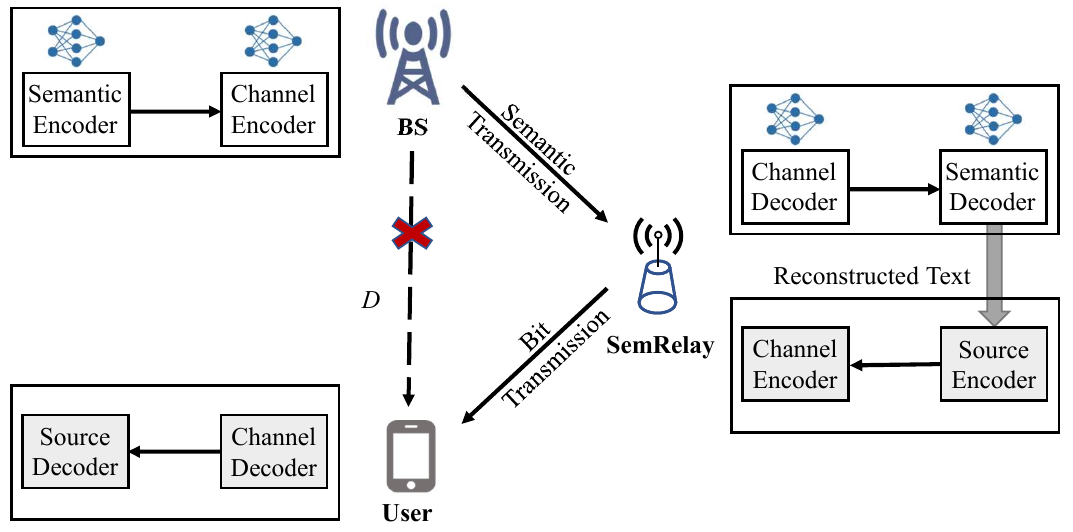}
	\caption{A SemRelay-aided text transmission system.}
	\label{max_rate_comparison}
\end{figure}

To address this issue, we, for the first time, propose a new \emph{semantic relay} (SemRelay) as illustrated in Fig. 1, to assist text transmission from a base station (BS) equipped with a DeepSC transmitter to a resource-constrained user. Specifically, different from the conventional decode-and-forward (DF) relay designed for bit transmissions in both the BS$\rightarrow$relay and relay$\rightarrow$user links, the SemRelay equipped with a well-trained DeepSC receiver first decodes the semantic information sent by the BS, and then forwards it to the user by adopting conventional bit-based transmission. The key advantages are two folds. First, compared with conventional DF relay, the SemRelay greatly improves the text transmission efficiency since semantic communication is performed in the transmitter-SemRelay link. Second, for resource-constrained users, the SemRelay effectively relieves their computation and storage burden, since the computation-demanding semantic decoding is performed at the SemRelay.

In this paper, we consider a SemRelay-aided single-user text transmission system, where the BS$\rightarrow$SemRelay semantic transmission and SemRelay$\rightarrow$user bit transmission are operated in orthogonal frequency bands. An optimization problem is formulated to maximize the achievable (effective) bit rate by jointly designing the bandwidth allocation and SemRelay placement. This problem, however, is non-convex with intricately coupled variables, which renders the widely-used alternating optimization (AO) method inefficient. To tackle this difficulty, we propose an efficient penalty-based algorithm to obtain a high-quality suboptimal solution. Numerical results demonstrate the effectiveness of the proposed algorithm as well as significant rate performance gain of the proposed SemRelay over conventional DF relay. In particular, it is shown that the SemRelay should be placed closer to the user when the total bandwidth decreases and more bandwidth should be allocated to bit transmission for achieving near-optimal performance.

\section{System Model and Problem Formulation}

\subsection{System Model}

Consider a SemRelay-aided text transmission system as shown in Fig. \ref{max_rate_comparison}, where a BS sends text information to a typical user in a user zone\footnote{This work can be extended to the communication system with multiple users by designing resource allocation among them to maximize the minimum rate of all users, which is left for our future work.}, both equipped with a single antenna\footnote{This work can be easily extended to the case with a multi-antenna BS by designing the transmit beamforming of the BS to maximize the SNR, which does not affect the main results of this paper.}. We assume that the BS has a well-trained DeepSC transmitter, while the user has limited storage and computation capabilities, hence incapable of decoding the semantic information sent by the BS\footnote{For the mobile devices with abundant computation and storage resources, the BS can directly send semantic information to them.}. Moreover, the direct link between the BS and user is assumed to be blocked due to long distance and severe blockage. Thus, a SemRelay with a DeepSC receiver model and abundant computation resources is properly deployed to assist the text transmission in two phases. Specifically, in Phase 1, the BS sends the text information to the SemRelay via its DeepSC transmitter. Next, in Phase 2, the SemRelay first decodes text from received signals by employing the channel-and-semantic decoders, and then forwards it to the user by adopting conventional bit-based transmission. Moreover, we assume negligible semantic decoding delay at the SemRelay, thanks to its abundant computation resources. Furthermore, to avoid communication interference, we adopt the frequency division multiple access (FDMA) scheme for the SemRelay, where the BS$\rightarrow$SemRelay semantic transmission and SemRelay$\rightarrow$user bit transmission are operated in orthogonal frequency bands. Let $d_\text{br}$ denote the distance between the BS and SemRelay. Without loss of generality, we consider a Cartesian coordinate system where the BS, SemRelay, and user are located at $\boldsymbol{u}_\text{b}=(0,0,0)$, $\boldsymbol{u}_\text{r}=(d_{\text{br}}, 0, H)$, and $\boldsymbol{u}_\text{u}=(D, 0, 0)$, respectively, with the SemRelay deployed at an altitude of $H$ to establish light-of-sight (LoS) dominant channels with the BS and user.

\subsubsection{Model for BS$\rightarrow$SemRelay semantic transmission}
Under the LoS-dominant channel model, the channel from the BS to SemRelay, denoted by $h_{\text{br}}$, can be modeled as $h_{\text{br}}=\rho_0/(d_{\text{br}}^2+H^2)^{{\beta}/2}$, where $\rho_0$ denotes the channel power gain at a reference distance $d_0 = 1$ meter (m), and $\beta$ denotes the path loss exponent. Let $W$ denote the total bandwidth of the system and $\alpha_{\text{br}}$ the fraction of bandwidth allocated to the BS$\rightarrow$SemRelay link. Then, the received SNR at the SemRelay, $\gamma_\text{{br}}$ (in dB), is given by
\begin{align}
\gamma_\text{{br}}=10\log_{10}\left(\frac{P_{\text{b}}\rho_0}{(d_{\text{br}}^2+H^2)^{\beta/2}\alpha_{\text{br}}WN_0}\right), \label{SemSNRdB}
\end{align}where $P_{\text{b}}$ is the BS transmit power and $N_{\text{0}}$ is the power spectral density of the received white Gaussian noise.

Let $\boldsymbol{s}$ denote a sentence sent by the BS and $L$ the average number of words per sentence. Unlike conventional bit-based transmission scheme, semantic transmission processes each sentence by feeding it into a DeepSC transceiver to generate a semantic symbol vector $\boldsymbol{w} \in \mathbb{R}^{KL \times 1}$, where $K$ denotes the average number of semantic symbols for each word in the original sentence. Moreover, let $I$ denote the semantic units (suts) that characterizes the average amount of semantic information contained in each sentence. Then $I/L$ in suts/word represents the average amount of semantic information contained in each word for each sentence. To evaluate the performance of the DeepSC model, a new performance metric, called \emph{semantic similarity}, was proposed in \cite{Huiqiang_Xie} that measures the distance of semantic information between two sentences. Specifically, the semantic similarity was shown to be dependent on the received SNR $\gamma_\text{{br}}$ and $K$. According to \cite{Lei_Yan} and \cite{Xidong_Mu}, for any $K$, the semantic similarity function $\varepsilon\in[0,1]$ generally follows an `S' shape with respect to SNR. Thus, it can be approximated as a sigmoid function by using the generalized logistic regression method, which is given by \cite{Xidong_Mu}
\begin{align}
\varepsilon=a_1+\frac{a_2}{1+e^{-\left(c_1\gamma_\text{{br}}+c_2\right)}}.
\end{align}
Herein, $a_1$, $a_2$, $c_1$, and $c_2$ are constant coefficients dependent on $K$. In order to ensure the accuracy of the recovered data, $\varepsilon$ should be lower-bounded by a minimum semantic similarity $\bar{\varepsilon}$, i.e, $\varepsilon \geq \bar{\varepsilon}$ \cite{Lei_Yan}. Following the above semantic similarity model, the achievable \emph{semantic rate} (susts/s) in the BS$\rightarrow$SemRelay link is given by \cite{Xidong_Mu}
\begin{align}
R_{\rm br}^{(\rm sem)}=\frac{\alpha_{\text{br}}WI}{KL}\varepsilon.\label{SemComRate}
\end{align}
Let $\mu$ in bits/word represent the average number of bits contained in each word in text transmission. Then the achievable \emph{semantic-to-bit} rate bits/s (bps) can be obtained as \cite{Lei_Yan}
\begin{align}
R_{\rm br}^{(\rm bit)}= \mu\frac{R_{\rm br}^{(\rm sem)}}{I/L}=\frac{\mu\alpha_{\text{br}}W}{K}\varepsilon.\label{WordComRate}
\end{align}

\subsubsection{Model for SemRelay$\rightarrow$user bit transmission} At the SemRelay, it first decodes the semantic symbols from the BS with its DeepSC receiver and then transmits them to the user based on bit transmission. Let $d_{\text{ru}}$ denote the horizontal distance between the SemRelay and user, $P_{\text{r}}$ denote the transmit power of the SemRelay, and $\alpha_{\text{ru}}$ denote the portion of bandwidth allocated to bit transmission for the SemRelay$\rightarrow$user link. Then the SemRelay$\rightarrow$user channel power gain is $h_{\text{ru}}=\rho_0/(d_{\text{ru}}^2+H^2)^{{\beta}/2}$, and its corresponding achievable bit rate in bps is
\begin{align}
R_{\rm ru}^{(\rm bit)}=\alpha_{\text{ru}}W\log_2\left(1+\frac{P_{\text{r}}\rho_0}{(d_{\text{ru}}^2+H^2)^{{\beta}/2}\alpha_{\text{ru}}WN_0}\right).\label{BitComRate}
\end{align}

\subsection{Problem Formulation}
For the considered SemRelay-aided text transmission system, our target is to maximize its achievable (effective) bit rate by jointly optimizing the SemRelay placement $\boldsymbol{d}\triangleq\{d_{\text{br}}, d_{\text{ru}}\}$ and bandwidth allocation $\boldsymbol{\alpha}\triangleq\{\alpha_{\text{br}}, \alpha_{\text{ru}}\}$. As the system achievable bit rate is $\eta\triangleq\min\{R_{\rm br}^{(\rm bit)},R_{\rm ru}^{(\rm bit)}\}$, this optimization problem can be formulated as
\begin{subequations} \label{prob1}
\begin{align}
{\text{(P1)}}\max \limits_{\boldsymbol{d},\boldsymbol{\alpha},\gamma_{\text{br}},\eta} \;\;&\eta\nonumber\\
\textrm{s.t.}\quad\;\;
&\eta \le {\alpha_{\text{ru}}W}\log_2\left(1+\frac{P_{\text{r}}\rho_0}{(d_{\text{ru}}^2+H^2)^{\beta/2}\alpha_{\text{ru}}WN_0}\right), \label{SlaRateb}\\
&\eta \le \frac{\alpha_{\text{br}}W\mu}{K}\left(a_1+\frac{a_2}{1+e^{-\left(c_1\gamma_{\text{br}}+c_2\right)}}\right), \label{InforCausCons} \\
&a_1+\frac{a_2}{1+e^{-\left(c_1\gamma_{\text{br}}+c_2\right)}} \geq \bar{\varepsilon}, \label{MinSimiCons}\\
&\gamma_{\text{br}} = 10\log_{10}\left(\frac{P_{\text{b}}\rho_0}{(d_{\text{br}}^2+H^2)^{\beta/2}\alpha_{\text{br}}WN_0}\right),\label{gamma}\\
&\alpha_{\text{br}} \geq 0, \alpha_{\text{ru}} \geq 0, \label{MinBandCons}\\
&\alpha_{\text{br}}+\alpha_{\text{ru}} = 1, \label{TotBandCons}\\
&d_{\text{br}} \geq 0, d_{\text{ru}} \geq 0, \label{LocaCons} \\
&d_{\text{br}} + d_{\text{ru}} = D, \label{LocaConsTot}
\end{align}
\end{subequations}
where \eqref{SlaRateb} and \eqref{InforCausCons} are the rate constraints for the bit and semantic transmissions, respectively, \eqref{TotBandCons} is the total bandwidth constraint, and \eqref{MinSimiCons} is the minimum semantic similarity constraint. Note that the constraint \eqref{MinSimiCons} can be equivalently expressed as
\vspace{-0.1cm}
\begin{align}
\gamma_{\text{br}} \geq \frac{1}{c_1}\ln\left(\frac{\bar{\varepsilon}-a_1}{a_1+a_2-\bar{\varepsilon}}\right)-\frac{c_2}{c_1}.\label{LocaSlaSimiThresh}
\end{align}
In particular, combining \eqref{SemSNRdB} and \eqref{LocaSlaSimiThresh}, we can obtain that the required bandwidth for semantic transmission is upper-bounded by $\alpha_{\text{br}}W\le \frac{P_{\text{b}}\rho_0}{(d_{\text{br}}^2+H^2)^{\beta/2}N_0\exp\left(\frac{\ln(10)}{10c_1}\ln\left(\frac{\bar{\varepsilon}-a_1}{a_1+a_2-\bar{\varepsilon}}\right)-\frac{\ln(10)c_2}{10c_1}\right)}$. This indicates that semantic communication can be more favorable in the low-bandwidth region.

Problem (P1) is a non-convex optimization problem since the constraints in \eqref{SlaRateb}--\eqref{gamma} are non-convex with coupled optimization variables. Although the optimal solution to problem (P1) can be obtained by a two-dimensional exhaustive search, it incurs prohibitively high computation complexity that is unaffordable in practice. Moreover, it can be numerically shown that the widely-used alternating AO method cannot be directly used for solving this problem, since the optimization variables are strongly coupled in \eqref{SlaRateb}--\eqref{MinSimiCons}, which thus renders the AO method getting stuck at a low-quality solution. To address the above issues, we propose an efficient iterative algorithm in the next section to obtain a high-quality suboptimal solution to (P1).

%\vspace{-0.2cm}
\section{Penalty-based Algorithm}
In this section, we propose a two-layer penalty-based algorithm to sub-optimally solve problem (P1). Specifically, the inner layer solves a penalized optimization problem by applying the block coordinate descent (BCD) method given a fixed penalty coefficient, while the outer layer updates the penalty coefficient, until the convergence is achieved.
\subsection{Problem Reformulation}
First, to address the issue of coupled variables in constraints \eqref{SlaRateb}, \eqref{InforCausCons}, and \eqref{gamma}, we introduce a set of auxiliary variables $\boldsymbol{z}\triangleq\{\hat{d}_{\text{br}}, \hat{d}_{\text{ru}}, \hat{\alpha}_{\text{br}}, \hat{\alpha}_{\text{ru}}\}$ for the constraints \eqref{TotBandCons} and \eqref{LocaConsTot} and define
{\setlength\abovedisplayskip{1pt}
\setlength\belowdisplayskip{3pt}
\begin{align}
&\alpha_{\text{br}} = \hat{\alpha}_{\text{br}}, \alpha_{\text{ru}} = \hat{\alpha}_{\text{ru}},\label{PenaltyB}\\
&d_{\text{br}} = \hat{d}_{\text{br}}, d_{\text{ru}} = \hat{d}_{\text{ru}}.\label{PenaltyD}
\end{align}}Then the constraints \eqref{TotBandCons} and \eqref{LocaConsTot} can be respectively re-expressed as
{\setlength\abovedisplayskip{1pt}
\setlength\belowdisplayskip{3pt}
\begin{align}
\hat{\alpha}_{\text{br}}+\hat{\alpha}_{\text{ru}} = 1, \label{TotBandConsP}\\
\hat{d}_{\text{br}} + \hat{d}_{\text{ru}} = D. \label{LocaConsTotP}
\end{align}}By replacing constraints \eqref{TotBandCons} and \eqref{LocaConsTot} with \eqref{TotBandConsP} and \eqref{LocaConsTotP}, respectively, problem (P1) can be equivalently transformed as
{\setlength\abovedisplayskip{0.1pt}
\setlength\belowdisplayskip{4pt}
\begin{subequations} \label{prob2}
\begin{align}
{\text{(P2)}}\max \limits_{\boldsymbol{d},\boldsymbol{\alpha},\boldsymbol{z},\gamma_{\text{br}},\eta}\;\;&\eta \nonumber\\
\textrm{s.t.}\quad\;\;\;
&\eqref{SlaRateb}, \eqref{InforCausCons}, \eqref{gamma}, \eqref{MinBandCons}, \eqref{LocaCons}, \eqref{LocaSlaSimiThresh}-\eqref{LocaConsTotP}. \nonumber
\end{align}
\end{subequations}}Next, we penalize the equality constraints \eqref{PenaltyB} and \eqref{PenaltyD} by adding them as a penalty term in the objective function of problem (P2), hence yielding the following problem
\begin{subequations} \label{prob3}
\begin{align}
{\text{(P3)}}\max \limits_{\boldsymbol{d},\boldsymbol{\alpha},\boldsymbol{z},\gamma_{\text{br}},\eta}\;\;&\eta-\frac{1}{2\lambda}\{(\alpha_{\text{br}}-\hat{\alpha}_{\text{br}})^2+(\alpha_{\text{ru}}-\hat{\alpha}_{\text{ru}})^2\nonumber\\
&\qquad\qquad+\nu(d_{\text{br}}-\hat{d}_{\text{br}})^2+\nu(d_{\text{ru}}-\hat{d}_{\text{ru}})^2\} \nonumber \\
\textrm{s.t.}\quad\;\;\;
&\eqref{SlaRateb}, \eqref{InforCausCons}, \eqref{gamma}, \eqref{MinBandCons}, \eqref{LocaCons},  \eqref{LocaSlaSimiThresh}, \eqref{TotBandConsP}, \eqref{LocaConsTotP}, \nonumber
\end{align}
\end{subequations}
where $\lambda$ denotes the penalty coefficient used for penalizing the equality constraints in (P2), and $\nu$ is a weight for balancing the magnitudes of the bandwidth and distance penalties. By gradually decreasing the value of $\lambda$, we can obtain a solution to (P3) that satisfies all equality constraints in (P2) within a predefined accuracy\cite{Q_shi}.

However, given ${\lambda}>0$, problem (P3) is still a non-convex problem due to the non-convex constraints in \eqref{SlaRateb}, \eqref{InforCausCons}, and \eqref{gamma}. To address this issue, we apply the BCD method to partition its optimization variables into three blocks: 1) SemRelay placement $\boldsymbol{d}$, 2) bandwidth allocation $\boldsymbol{\alpha}$, and 3) auxiliary variables $\boldsymbol{z}$. Then, the above three blocks are optimized alternately and iteratively until the convergence is reached.

\subsection{Inner Layer: BCD Method for Solving \text{(P3)}}

\subsubsection{Placement optimization}
For any given bandwidth allocation $\boldsymbol{\alpha}$ and auxiliary variables $\boldsymbol{z}$, problem (P3) reduces to the following problem for optimizing the SemRelay placement
\begin{subequations} \label{Locaprob1}
\begin{align}
{\text{(P4)}}\max \limits_{\boldsymbol{d},\gamma_{\text{br}},\eta} \;\;&\eta-\frac{\nu}{2\lambda}\{(d_{\text{br}} - \hat{d}_{\text{br}})^2+(d_{\text{ru}} - \hat{d}_{\text{ru}})^2\} \nonumber \\
\textrm{s.t.}\quad&\eqref{SlaRateb}, \eqref{InforCausCons}, \eqref{gamma}, \eqref{LocaCons}, \eqref{LocaSlaSimiThresh}.\nonumber
\end{align}
\end{subequations}Problem (P4) is a non-convex optimization problem due to the non-convex constraints in \eqref{SlaRateb}, \eqref{InforCausCons}, and \eqref{gamma}. Fortunately, they can be efficiently addressed by using the successive convex approximation (SCA) method, as elaborated below.
\begin{lemma}
\emph{
For the constraint \eqref{SlaRateb}, $R_{\rm ru}^{(\rm bit)}$ is a convex function of $(d_{\text{ru}}^2+H^2)^{\frac{\beta}{2}}$. For any local point $\tilde{d}_{\text{ru}}$, $R_{\rm ru}^{(\rm bit)}$ in \eqref{SlaRateb} can be lower-bounded as
\begin{align}
R_{\rm ru}^{(\rm bit)} &\geq {\alpha_{\text{ru}}W}\left(E_1-E_2\left((d_{\text{ru}}^2+H^2)^{\frac{\beta}{2}}-(\tilde{d}_{\text{ru}}^2+H^2)^{\frac{\beta}{2}}\right)\right)\nonumber\\
&\triangleq R_{\rm ru}^{(\rm bit)(\rm lb)}, \label{LocaSlaRateb1}
\end{align}
where the coefficients $E_1$ and $E_2$ are respectively given by $E_1=\log_2\left(1+\frac{P_{\text{r}}\rho_0}{(\tilde{d}_{\text{ru}}^2+H^2)^{{\beta}/2}\alpha_{\text{ru}}WN_0}\right)$ and $E_2= \left(\frac{P_{\text{r}}\rho_0\log_2(e)}{(\tilde{d}_{\text{ru}}^2+H^2)^{{\beta}}\alpha_{\text{ru}}WN_0}\right)/\left(1+\frac{P_{\text{r}}\rho_0}{(\tilde{d}_{\text{ru}}^2+H^2)^{{\beta}/2}\alpha_{\text{ru}}WN_0}\right)$.
}
\end{lemma}
\begin{proof}
It can be proved by applying the first-order Taylor expansion for a convex function to obtain its lower bound \cite{S_Boyd}.
\end{proof}
\begin{lemma}
\emph{
For the constraint \eqref{InforCausCons}, $\psi\triangleq1/\left(1+e^{-\left(c_1\gamma_{\text{br}}+c_2\right)}\right)$ is a convex function of $\left(1+e^{-\left(c_1\gamma_{\text{br}}+c_2\right)}\right)$. For any local point $\tilde{\gamma}_{\text{br}}$, $\psi$ can be lower-bounded as
\begin{align}
\psi\geq E_3-E_4\left(e^{-\chi}-e^{-\tilde{\chi}}\right)\triangleq{\psi}^{(\rm lb)},\label{PlaSRRateCons1}
\end{align}where $\chi=c_1\gamma_{\text{br}}+c_2$, and the coefficients $\tilde{\chi}$, $E_3$, and $E_4$ are given by $\tilde{\chi}=c_1\tilde{\gamma}_{\text{br}}+c_2$, $E_3 = 1/\left(1+e^{-\left(c_1\tilde{\gamma}_{\text{br}}+c_2\right)}\right)$, and $E_4 = 1/\left(1+e^{-\left(c_1\tilde{\gamma}_{\text{br}}+c_2\right)}\right)^2$, respectively.
}
\end{lemma}

Next, for the non-affine constraint \eqref{gamma}, we first relax it as
\begin{align}
\gamma_{\text{br}}&\le10\log_{10}\left(\frac{P_{\text{b}}\rho_0}{(d_{\text{br}}^2+H^2)^{\beta/2}\alpha_{\text{br}}WN_0}\right)\nonumber\\
&=10\log_{10}\left(\frac{P_{\text{b}}\rho_0}{\alpha_{\text{br}}WN_0}\right)-5\beta\log_{10}\left(d_{\text{br}}^2+H^2\right). \label{Plagamma}
\end{align}
Then, it can be proved by contradiction that the equality in the constraint \eqref{Plagamma} always holds in the optimal solution to the relaxed problem. Moreover, we have the following result.
\begin{lemma}
\emph{
For the constraint \eqref{Plagamma}, $\phi\triangleq\log_{10}\left(d_{\text{br}}^2+H^2\right)$ is a concave function of $d_{\text{br}}^2$. For any local point $\tilde{d}_{\text{br}}$, $\phi$ can be upper-bounded as
\begin{align}
\phi\le E_5+E_6(d_{\text{br}}^2-\tilde{d}_{\text{br}}^2)\triangleq \phi^{(\rm up)},\label{GammaCons1}
\end{align}where the coefficients $E_5$ and $E_6$ are respectively given by $E_5 = \log_{10}(\tilde{d}_{\text{br}}^2+H^2)$ and $E_6=\log_{10}(e)/(\tilde{d}_{\text{br}}^2+H^2)$.
}
\end{lemma}

Based on Lemmas 1--3, by replacing $R_{\rm ru}^{(\rm bit)}$ in \eqref{SlaRateb}, $\psi$ in \eqref{InforCausCons}, and $\phi$ in \eqref{Plagamma} with their corresponding lower or upper bounds in \eqref{LocaSlaRateb1}, \eqref{PlaSRRateCons1}, and \eqref{GammaCons1}, respectively, problem (P4) can be transformed into the following approximate form
\begin{subequations} \label{Locaprob2}
\begin{align}
{\text{(P5)}}\max \limits_{\boldsymbol{d},\gamma_{\text{br}},\eta} \;\;&\eta-\frac{\nu}{2\lambda}\{(d_{\text{br}} - \hat{d}_{\text{br}})^2+(d_{\text{ru}} - \hat{d}_{\text{ru}})^2\} \nonumber \\
\textrm{s.t.}\quad
&\eta \le R_{\rm ru}^{(\rm bit)(\rm lb)},\label{RruOpt}\\
&\eta \le \frac{\alpha_{\text{br}}W\mu}{K}\left(a_1+a_2\psi^{(\rm lb)}\right),\label{RbrOpt}\\
&\gamma_{\text{br}}\le10\log_{10}\left(\frac{P_{\text{b}}\rho_0}{\alpha_{\text{br}}WN_0}\right)-5\beta\phi^{(\rm up)},\label{GammabrOpt}\\
&\eqref{LocaCons}, \eqref{LocaSlaSimiThresh}. \nonumber
\end{align}
\end{subequations}
Problem (P5) is now a convex optimization problem, which can be efficiently solved by using CVX solvers \cite{CVX}.

\subsubsection{Bandwidth optimization}
For any given SemRelay placement $\boldsymbol{d}$ and auxiliary variables $\boldsymbol{z}$, problem (P3) reduces to the following problem for optimizing the bandwidth allocation
\begin{subequations} \label{Bandprob1}
\begin{align}
{\text{(P6)}}\max \limits_{\boldsymbol{\alpha},\gamma_{\text{br}},\eta} \;\;&\eta-\frac{1}{2\lambda}\{(\alpha_{\text{br}} - \hat{\alpha}_{\text{br}})^2+(\alpha_{\text{ru}} - \hat{\alpha}_{\text{ru}})^2\}\nonumber\\
\textrm{s.t.}\quad
&\eqref{SlaRateb}, \eqref{InforCausCons}, \eqref{gamma}, \eqref{MinBandCons}, \eqref{LocaSlaSimiThresh}.\nonumber
\end{align}
\end{subequations}
Problem (P6) is a non-convex optimization problem due to the non-convex constraints in \eqref{InforCausCons} and \eqref{gamma}. To address this issue, we first equivalently rewrite the constraint \eqref{InforCausCons} as
\begin{align}
\eta &\le \frac{\alpha_{\text{br}}SW\mu}{K}=\frac{W\mu}{4K}\left((\alpha_{\text{br}}+S)^2-(\alpha_{\text{br}}-S)^2\right), \label{BandInforCausCons}\\
S&=a_1+\frac{a_2}{1+e^{-\left(c_1\gamma_{\text{br}}+c_2\right)}}.\label{BandSlacSimiequ}
\end{align}
Although \eqref{BandInforCausCons} is still a non-convex constraint, we can address it by using the SCA method, as shown below.

\begin{lemma}
\emph{
For the constraint \eqref{BandInforCausCons}, $\delta\triangleq(\alpha_{\text{br}}+S)^2$ is a convex function of $\alpha_{\text{br}}$ and $S$. For any local points $\tilde{\alpha}_{\text{br}}$ and $\tilde{S}$, $\delta$ can be lower-bounded as
\begin{align}
\delta&\geq-(\tilde{\alpha}_{\text{br}}+\tilde{S})^2+2(\tilde{\alpha}_{\text{br}}+\tilde{S})(\alpha_{\text{br}}+S)\triangleq \delta^{(\rm lb)}. \label{BandInforCausCons2}
\end{align}
}
\end{lemma}
\noindent Using Lemma 4, the constraint \eqref{BandInforCausCons} can be upper-bounded as
\begin{align}
\eta\le\frac{W\mu}{4K}\left(\delta^{(\rm lb)}-(\alpha_{\text{br}}-S)^2\right).\label{BandInforCausCons2}
\end{align}
Next, for the non-affine constraint \eqref{BandSlacSimiequ}, we first relax it as
\begin{align}
S\le a_1+\frac{a_2}{1+e^{-\left(c_1\gamma_{\text{br}}+c_2\right)}}. \label{BandSlacSimi}
\end{align}
Then we can easily show that the relaxation actually does not affect the optimality since the equality in the constraint \eqref{BandSlacSimi} holds in the optimal solution to the relaxed problem.

Moreover, one can observe that, given $\boldsymbol{d}$ and $\boldsymbol{z}$, the non-convex constraints \eqref{BandSlacSimi} and \eqref{gamma} in (P6) have similar forms with the constraints \eqref{InforCausCons} and \eqref{gamma} in (P4), respectively. Following the similar procedures based on the SCA method, \eqref{BandSlacSimi} and \eqref{gamma} in problem (P6) can be transformed as
\begin{align}
&S\le a_1+a_2\left(E_7-E_8\left(e^{-\tau}-e^{-\tilde{\tau}}\right)\right),\label{BandMinSimiCons1}\\
&\gamma_{\text{br}} \le 10\log_{10}\left(\frac{P_{\text{b}}\rho_0}{(d_{\text{br}}^2+H^2)^{\beta/2}WN_0}\right)-E_9-E_{10}(\alpha_{\text{br}}-\tilde{\alpha}_{\text{br}}),\label{BandGammaCons1}
%&\resizebox{.55\hsize}{!}{$S\le a_1+a_2\left(E_7-E_8\left(e^{-\tau}-e^{-\tilde{\tau}}\right)\right),$}\label{BandMinSimiCons1}\\
%&\resizebox{.897\hsize}{!}{$\gamma_{\text{br}} \le 10\log_{10}\left(\frac{P_{\text{b}}\rho_0}{(d_{\text{br}}^2+H^2)^{\beta/2}WN_0}\right)-E_9-E_{10}(\alpha_{\text{br}}-\tilde{\alpha}_{\text{br}}),$}\label{BandGammaCons1}
%\gamma_{\text{br}} \le 10\log_{10}\left(\frac{P_{\text{b}}\rho_0}{(d_{\text{br}}^2+H^2)^{\beta/2}WN_0}\right)-E_9-E_{10}(\alpha_{\text{br}}-\tilde{\alpha}_{\text{br}}),
\end{align}
where $\tau=c_1\gamma_{\text{br}}+c_2$, and the coefficients $\tilde{\tau}=c_1\tilde{\gamma}_{\text{br}}+c_2$, $E_7 = 1/\left(1+e^{-\left(c_1\tilde{\gamma}_{\text{br}}+c_2\right)}\right)$, $E_8 = 1/\left(1+e^{-\left(c_1\tilde{\gamma}_{\text{br}}+c_2\right)}\right)^2$,  $E_9=10\log_{10}(\tilde{\alpha}_{\text{br}})$, and $E_{10}=10\log_{10}(e)/\tilde{\alpha}_{\text{br}}$ are determined by the local point $\tilde{\gamma}_{\text{br}}$ and $\tilde{\alpha}_{\text{br}}$, respectively.

By replacing the constraints \eqref{InforCausCons} and \eqref{gamma} in problem (P6) with those in \eqref{BandInforCausCons2}, \eqref{BandMinSimiCons1}, and \eqref{BandGammaCons1}, problem (P6) can be transformed into the following approximate problem
\begin{subequations} \label{Bandprob2}
\begin{align}
{\text{(P7)}}\max_{\tiny
     \begin{array}{cc}
         \boldsymbol{\alpha},\gamma_{\text{br}},S,\eta
     \end{array}
     }\;&\eta-\frac{1}{2\lambda}\{(\alpha_{\text{br}} - \hat{\alpha}_{\text{br}})^2+(\alpha_{\text{ru}} - \hat{\alpha}_{\text{ru}})^2\}  \nonumber\\
\textrm{s.t.}\quad\;\;
&\eqref{SlaRateb}, \eqref{MinBandCons}, \eqref{LocaSlaSimiThresh}, \eqref{BandInforCausCons2}, \eqref{BandMinSimiCons1}, \eqref{BandGammaCons1}.\nonumber
\end{align}
\end{subequations}
Problem (P7) is now a convex optimization problem, which can be efficiently solved by using CVX solvers \cite{CVX}.
\subsubsection{Auxiliary variables optimization}
For any given the SemRelay placement $\boldsymbol{d}$ and bandwidth allocation $\boldsymbol{\alpha}$, problem (P3) reduces to
\begin{subequations} \label{Slacprob1}
\begin{align}
{\text{(P8)}}\max \limits_{\boldsymbol{z}} \;\;&-\frac{1}{2\lambda}\{(\alpha_{\text{br}}-\hat{\alpha}_{\text{br}})^2+(\alpha_{\text{ru}}-\hat{\alpha}_{\text{ru}})^2\nonumber\\
&\qquad\qquad+\nu(d_{\text{br}}-\hat{d}_{\text{br}})^2+\nu(d_{\text{ru}}-\hat{d}_{\text{ru}})^2\}  \nonumber\\
\textrm{s.t.}\;\;
&\eqref{TotBandConsP}, \eqref{LocaConsTotP}. \nonumber
\end{align}
\end{subequations}
Problem (P8) is a convex problem which can be efficiently solved by using CVX solvers \cite{CVX}.

To summarize, a suboptimal solution to (P3) can be obtained by solving problems (P5), (P7), and (P8) alternately and iteratively, until the convergence is achieved.

\subsection{Outer Layer: Update Penalty Coefficient}

The outer layer aims to update the penalty coefficient $\lambda$ with $\lambda=c\lambda$ to guarantee the converged solution satisfying the equality constraints in (P2), where $c$ is a constant scaling factor.

\subsection{Algorithm Analysis}

For any solution obtained to (P3), we evaluate the violation of equality constraints in (P2) by adopting the maximum of four penalty terms, $\zeta$. The proposed algorithm is terminated when $\zeta$ is less than a predefined accuracy $\epsilon_{\rm 1}$ for all equality constraints. Note that as the penalty coefficient $\lambda$ decreases, the penalty term becomes larger and eventually guarantees the equality constraints. Moreover, for any $\lambda$, the objective value of (P3) in the inner layer can be obtained by solving convex subproblems (P5), (P7), and (P8) alternately and iteratively, which thus is non-decreasing over iterations. Besides, since the objective value is upper-bounded, the proposed penalty-based algorithm is guaranteed to converge \cite{Q_shi}. Furthermore, the overall algorithm complexity can be characterized as $\mathcal{O}(I_{\text{in}}I_{\text{out}}N^{3.5})$, where $N$ is the total number of optimization variables, $I_{\text{in}}$ and $I_{\text{out}}$ denote the numbers of the outer and inner iterations required for convergence, respectively \cite{S_Boyd}.

%width=1\linewidth, height=0.75\linewidth
\section{Numerical Results}
%We first examine the conventional resource allocation model in semantic-aware networks. 需要考虑传统的资源分配在semantic transmission 系统下的性能吗？这样的对比方案的意义是什么呢？参考秦志金资源分配短文仿真结果。
In this section, we present numerical results to evaluate the effectiveness of the proposed SemRelay as well as the penalty-based algorithm. The simulation setup is as follows without otherwise specified. The BS-user horizontal distance is $D=100$ m and the relay is deployed at an altitude of $H=10$ m. The reference path loss is $\rho_0=-60$ dB and the path loss exponent is set as $\beta=3$. According to \cite{Xidong_Mu}, we set $K=4$, for which we have $a_1=0.3980$, $a_2 = 0.5385$, $c_1 = 0.2815$, and $c_2 = -1.3135$. Moreover, we assume that, on average, each word contains five letters and each letter is encoded using the ASCll code, thus we have $\mu=40$ bits/word \cite{Lei_Yan}. Other parameters are set as $\lambda=1000$, $c=0.9$, $\nu=10^{-4}$, $\bar{\varepsilon}=0.9$, $P_{\text{b}}=P_{\text{r}}=0.1$ W, $N_0=-169$ dBm/Hz, and $\epsilon_{\rm 1}=10^{-8}$. Furthermore, we compare the performance of the proposed algorithm with the following benchmark schemes: 1) SemRelay design based on the exhaustive search, 2) SemRelay with optimized placement given equal bandwidth allocation, 3) SemRelay with optimized bandwidth allocation given fixed placement at the middle between the BS and user, and 4) conventional DF relay with joint placement optimization and bandwidth allocation.

Fig. \ref{rate_with_totband} shows the achievable bit rate obtained by different schemes versus the total bandwidth $W$. First, it is observed that for the SemRelay-aided text transmission system, our proposed penalty-based algorithm achieves close-to-optimal performance with the one based on the exhaustive search. Second, the proposed SemRelay-aided system achieves a much higher bit rate than the conventional DF relay when $W$ is small, since the BS$\rightarrow$user semantic transmission link transmits only the semantic information extracted from the original text, hence achieving higher spectral efficiency. Therefore, in the case of limited bandwidth, there are more degrees-of-freedom in the bandwidth design for further improving the performance of the SemRelay-aided communication system. On the other hand, when $W$ is sufficiently large, the conventional DF relay tends to provide a larger effective bit rate than the SemRelay. This can be intuitively understood, since the required bandwidth of semantic transmission is upper-bounded for satisfying the semantic similarity constraint (see \eqref{MinSimiCons} and \eqref{LocaSlaSimiThresh}), hence leading to a slower rate increasing. Besides, the proposed penalty-based algorithm significantly outperforms the benchmark schemes that optimize either bandwidth allocation or SemRelay placement, while the scheme with optimized bandwidth allocation achieves better performance than the one with optimized placement. In addition, the achievable rate obtained by the SemRelay with optimized bandwidth allocation only does not improve when the bandwidth is large, since given fixed SemRelay placement, the bandwidth allocated to the BS$\rightarrow$SemRelay link is upper-bounded (see Fig. \ref{bandasr_with_B}) for satisfying the semantic similarity requirement.

\begin{figure}[!t]%% 图,[htbp] 是浮动格式
\centering
\begin{minipage}[t]{1\linewidth}  \label{Fig.4}      % 图片占用一行宽度的30%
\centering
\includegraphics[width=8.5cm, height=6.5cm]{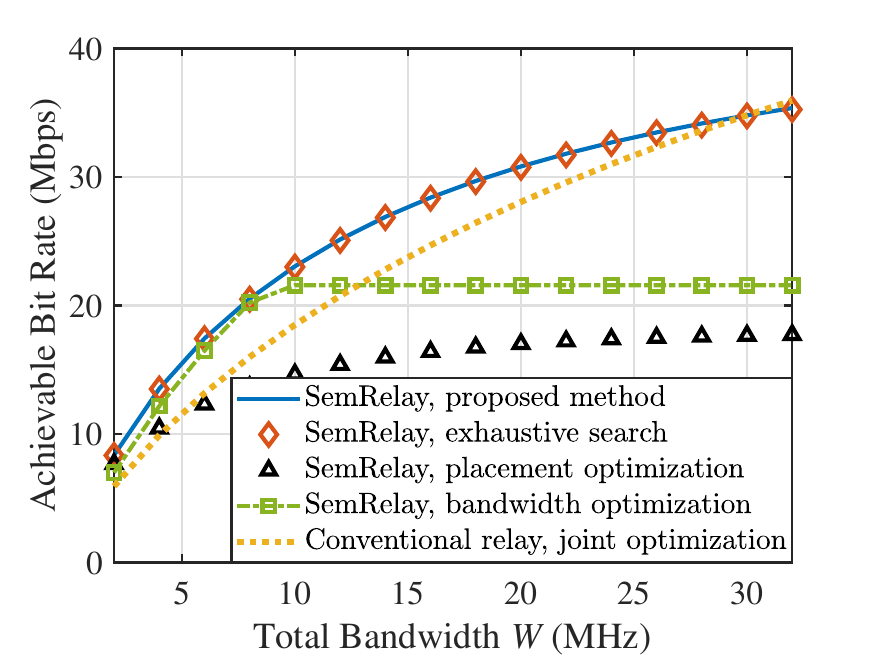}%7.5 5.7
\caption{The achievable bit rate versus total bandwidth.}
\label{rate_with_totband}
\end{minipage}
\begin{minipage}[t]{1\linewidth}  \label{Fig.4}      % 图片占用一行宽度的30%
\centering
\includegraphics[width=8.5cm, height=6.5cm]{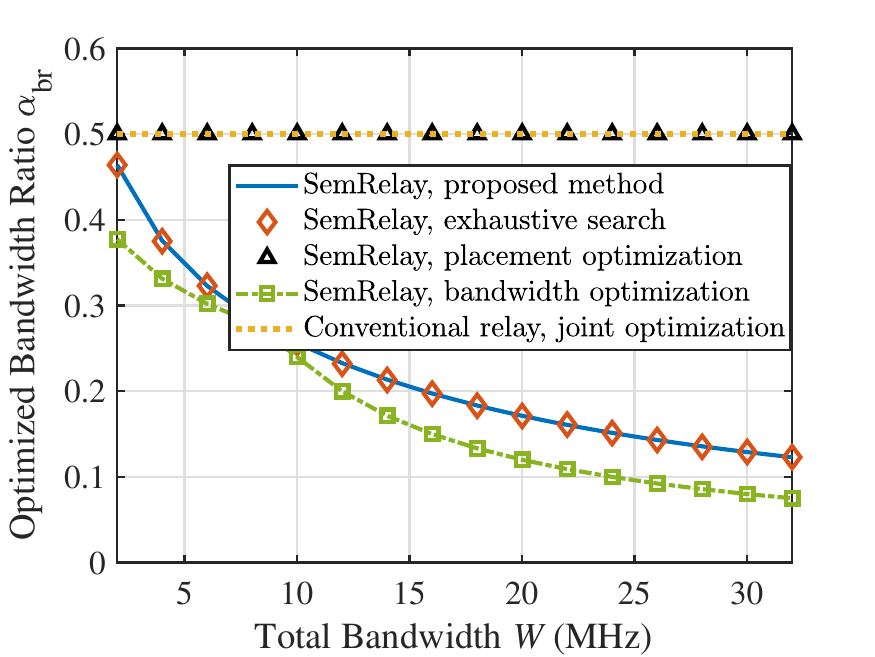}
\caption{Optimized bandwidth allocation versus total bandwidth.}
\label{bandasr_with_B}
\end{minipage}
\begin{minipage}[t]{1\linewidth}  \label{Fig.4}      % 图片占用一行宽度的30%
\centering
\includegraphics[width=8.5cm, height=6.5cm]{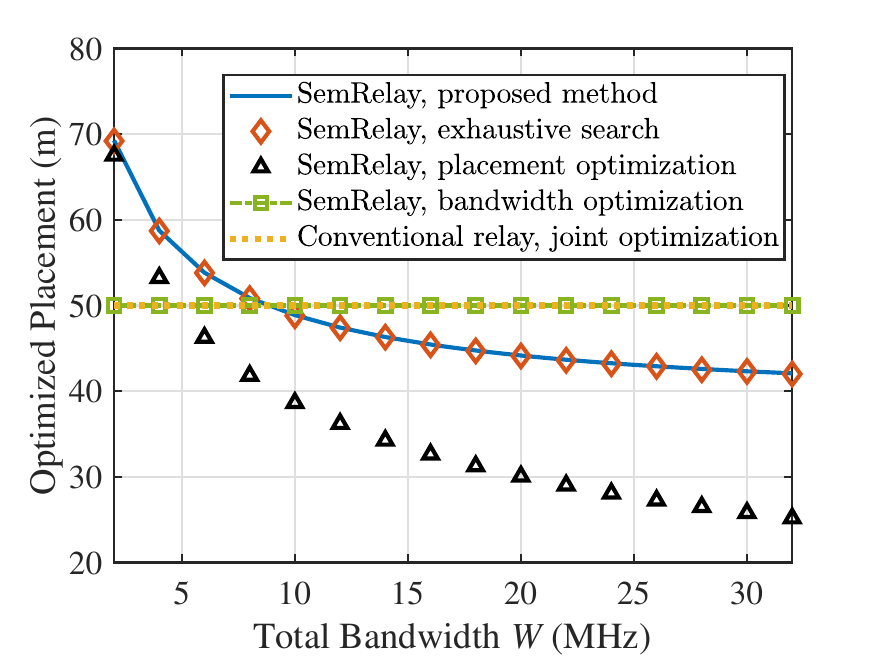}
\caption{Optimized SemRelay placement versus total bandwidth.}
\label{optimal_placement_with_B}
\end{minipage}
\end{figure}

Fig. \ref{bandasr_with_B} and Fig. \ref{optimal_placement_with_B} respectively show the optimized bandwidth allocation and relay placement of different schemes versus the total bandwidth $W$. Specifically, in Fig. \ref{bandasr_with_B}, unlike the conventional DF relay system with equal bandwidth allocation, in the SemRelay-aided system, the BS$\rightarrow$SemRelay link is allocated with less bandwidth than the SemRelay$\rightarrow$user link, since semantic communication can achieve much higher spectrum efficiency than the conventional bit transmission as only compressed semantic information need to be transmitted. In Fig. \ref{optimal_placement_with_B}, it is observed that for facilitating downlink communication, the SemRelay should be placed closer to the user when the total bandwidth decreases. This is expected since with less bandwidth, the performance of the SemRelay$\rightarrow$user bit transmission link, which has lower spectral efficiency than the BS$\rightarrow$SemRelay link, needs to be improved by reducing path loss.

%\vspace{-0.2cm}
\section{Conclusions}
In this paper, we proposed a new SemRelay-aided text transmission system, where a SemRelay equipped with the DeepSC receiver is deployed to first decode the semantic information sent by a resource-abundant BS, and then forwards it to the resource-constrained user based on conventional bit transmission. An optimization problem was formulated to maximize the achievable (effective) bit rate by jointly designing the SemRelay placement and bandwidth allocation, which was efficiently solved by using a penalty-based algorithm. Numerical results demonstrated close-to-optimal performance of the proposed algorithm as well as significant rate performance gain of the proposed SemRelay over conventional DF relay.

\section{Acknowledgment}
This work was supported by the National Natural Science Foundation of China under Grant 62201242, 62331023, and Young Elite Scientists Sponsorship Program by CAST 2022QNRC001.

\end{document}

%% file: header.tex
\newtheorem{theorem}{\emph{\underline{Theorem}}}
\newtheorem{acknowledgement}[theorem]{Acknowledgement}
\renewcommand{\algorithmicensure}{ \textbf{repeat:}}
\newtheorem{axiom}[theorem]{Axiom}
\newtheorem{case}[theorem]{Case}
\newtheorem{claim}[theorem]{Claim}
\newtheorem{conclusion}[theorem]{Conclusion}
\newtheorem{condition}[theorem]{Condition}
\newtheorem{conjecture}[theorem]{\emph{\underline{Conjecture}}}
\newtheorem{criterion}[theorem]{Criterion}
\newtheorem{definition}{\emph{\underline{Definition}}}
\newtheorem{exercise}[theorem]{Exercise}
\newtheorem{lemma}{\emph{\underline{Lemma}}}
\newtheorem{corollary}{\emph{\underline{Corollary}}}
\newtheorem{notation}[theorem]{Notation}
\newtheorem{problem}[theorem]{Problem}
\newtheorem{proposition}{\emph{\underline{Proposition}}}
\newtheorem{solution}[theorem]{Solution}
\newtheorem{summary}[theorem]{Summary}
\newtheorem{assumption}{Assumption}
\newtheorem{example}{\bf Example}
\newtheorem{remark}{\bf \emph{\underline{Remark}}}

\def\qed{$\Box$}
\def\QED{\mbox{\phantom{m}}\nolinebreak\hfill$\,\Box$}
\def\proof{\noindent{\emph{Proof:} }}
\def\poof{\noindent{\emph{Sketch of Proof:} }}
\def
\endproof{\hspace*{\fill}~\qed
\par
\endtrivlist\unskip}
\def\endproof{\hspace*{\fill}~\qed\par\endtrivlist\vskip3pt}

\def\E{\mathsf{E}}
\def\eps{\varepsilon}
\def\Lsp{{\boldsymbol L}}
\def\Bsp{{\boldsymbol B}}
\def\lsp{{\boldsymbol\ell}}
\def\Ltsp{{\Lsp^2}}
\def\Lpsp{{\Lsp^p}}
\def\Linsp{{\Lsp^{\infty}}}
\def\LtR{{\Lsp^2(\Rst)}}
\def\ltZ{{\lsp^2(\Zst)}}
\def\ltsp{{\lsp^2}}
\def\ltZt{{\lsp^2(\Zst^{2})}}
\def\ninN{{n{\in}\Nst}}
\def\oh{{\frac{1}{2}}}
\def\grass{{\cal G}}
\def\ord{{\cal O}}
\def\dist{{d_G}}
\def\conj#1{{\overline#1}}
\def\ntoinf{{n \rightarrow \infty }}
\def\toinf{{\rightarrow \infty }}
\def\tozero{{\rightarrow 0 }}
\def\trace{{\operatorname{trace}}}
\def\ord{{\cal O}}
\def\UU{{\cal U}}
\def\rank{{\operatorname{rank}}}
\def\acos{{\operatorname{acos}}}

\def\SINR{\mathsf{SINR}}
\def\SNR{\mathsf{SNR}}
\def\SIR{\mathsf{SIR}}
\def\tSIR{\widetilde{\mathsf{SIR}}}
\def\Ei{\mathsf{Ei}}
\def\l{\left}
\def\r{\right}
\def\({\left(}
\def\){\right)}
\def\lb{\left\{}
\def\rb{\right\}}

\setcounter{page}{1}

% Definitions
\newcommand{\eref}[1]{(\ref{#1})}
\newcommand{\fig}[1]{Fig.\ \ref{#1}}

% Bold lowercase
\def\bydef{:=}
\def\ba{{\mathbf{a}}}
\def\bb{{\mathbf{b}}}
\def\bc{{\mathbf{c}}}
\def\bd{{\mathbf{d}}}
\def\bee{{\mathbf{e}}}
\def\bff{{\mathbf{f}}}
\def\bg{{\mathbf{g}}}
\def\bh{{\mathbf{h}}}
\def\bi{{\mathbf{i}}}
\def\bj{{\mathbf{j}}}
\def\bk{{\mathbf{k}}}
\def\bl{{\mathbf{l}}}
\def\bm{{\mathbf{m}}}
\def\bn{{\mathbf{n}}}
\def\bo{{\mathbf{o}}}
\def\bp{{\mathbf{p}}}
\def\bq{{\mathbf{q}}}
\def\br{{\mathbf{r}}}
\def\bs{{\mathbf{s}}}
\def\bt{{\mathbf{t}}}
\def\bu{{\mathbf{u}}}
\def\bv{{\mathbf{v}}}
\def\bw{{\mathbf{w}}}
\def\bx{{\mathbf{x}}}
\def\by{{\mathbf{y}}}
\def\bz{{\mathbf{z}}}
\def\b0{{\mathbf{0}}}

% Bold capital letters
\def\bA{{\mathbf{A}}}
\def\bB{{\mathbf{B}}}
\def\bC{{\mathbf{C}}}
\def\bD{{\mathbf{D}}}
\def\bE{{\mathbf{E}}}
\def\bF{{\mathbf{F}}}
\def\bG{{\mathbf{G}}}
\def\bH{{\mathbf{H}}}
\def\bI{{\mathbf{I}}}
\def\bJ{{\mathbf{J}}}
\def\bK{{\mathbf{K}}}
\def\bL{{\mathbf{L}}}
\def\bM{{\mathbf{M}}}
\def\bN{{\mathbf{N}}}
\def\bO{{\mathbf{O}}}
\def\bP{{\mathbf{P}}}
\def\bQ{{\mathbf{Q}}}
\def\bR{{\mathbf{R}}}
\def\bS{{\mathbf{S}}}
\def\bT{{\mathbf{T}}}
\def\bU{{\mathbf{U}}}
\def\bV{{\mathbf{V}}}
\def\bW{{\mathbf{W}}}
\def\bX{{\mathbf{X}}}
\def\bY{{\mathbf{Y}}}
\def\bZ{{\mathbf{Z}}}

% mathbb Bold capital letters
\def\mA{{\mathbb{A}}}
\def\mB{{\mathbb{B}}}
\def\mC{{\mathbb{C}}}
\def\mD{{\mathbb{D}}}
\def\mE{{\mathbb{E}}}
\def\mF{{\mathbb{F}}}
\def\mG{{\mathbb{G}}}
\def\mH{{\mathbb{H}}}
\def\mI{{\mathbb{I}}}
\def\mJ{{\mathbb{J}}}
\def\mK{{\mathbb{K}}}
\def\mL{{\mathbb{L}}}
\def\mM{{\mathbb{M}}}
\def\mN{{\mathbb{N}}}
\def\mO{{\mathbb{O}}}
\def\mP{{\mathbb{P}}}
\def\mQ{{\mathbb{Q}}}
\def\mR{{\mathbb{R}}}
\def\mS{{\mathbb{S}}}
\def\mT{{\mathbb{T}}}
\def\mU{{\mathbb{U}}}
\def\mV{{\mathbb{V}}}
\def\mW{{\mathbb{W}}}
\def\mX{{\mathbb{X}}}
\def\mY{{\mathbb{Y}}}
\def\mZ{{\mathbb{Z}}}

% Caligraphic capital letters
\def\cA{\mathcal{A}}
\def\cB{\mathcal{B}}
\def\cC{\mathcal{C}}
\def\cD{\mathcal{D}}
\def\cE{\mathcal{E}}
\def\cF{\mathcal{F}}
\def\cG{\mathcal{G}}
\def\cH{\mathcal{H}}
\def\cI{\mathcal{I}}
\def\cJ{\mathcal{J}}
\def\cK{\mathcal{K}}
\def\cL{\mathcal{L}}
\def\cM{\mathcal{M}}
\def\cN{\mathcal{N}}
\def\cO{\mathcal{O}}
\def\cP{\mathcal{P}}
\def\cQ{\mathcal{Q}}
\def\cR{\mathcal{R}}
\def\cS{\mathcal{S}}
\def\cT{\mathcal{T}}
\def\cU{\mathcal{U}}
\def\cV{\mathcal{V}}
\def\cW{\mathcal{W}}
\def\cX{\mathcal{X}}
\def\cY{\mathcal{Y}}
\def\cZ{\mathcal{Z}}
\def\cd{\mathcal{d}}
\def\Mt{M_{t}}
\def\Mr{M_{r}}
%% my defs
\def\O{\Omega_{M_{t}}}
\newcommand{\figref}[1]{{Fig.}~\ref{#1}}
\newcommand{\tabref}[1]{{Table}~\ref{#1}}

%% From Kaibin
\newcommand{\var}{\mathsf{var}}
\newcommand{\fb}{\tx{fb}}
\newcommand{\nf}{\tx{nf}}
\newcommand{\BC}{\tx{(bc)}}
\newcommand{\MAC}{\tx{(mac)}}
\newcommand{\Pout}{p_{\mathsf{out}}}
\newcommand{\nnn}{\nn\\}
\newcommand{\FB}{\tx{FB}}
\newcommand{\TX}{\tx{TX}}
\newcommand{\RX}{\tx{RX}}
\renewcommand{\mod}{\tx{mod}}
\newcommand{\m}[1]{\mathbf{#1}}
\newcommand{\td}[1]{\tilde{#1}}
\newcommand{\sbf}[1]{\scriptsize{\textbf{#1}}}
\newcommand{\stxt}[1]{\scriptsize{\textrm{#1}}}
\newcommand{\suml}[2]{\sum\limits_{#1}^{#2}}
\newcommand{\sumlk}{\sum\limits_{k=0}^{K-1}}
\newcommand{\eqhsp}{\hspace{10 pt}}
\newcommand{\tx}[1]{\texttt{#1}}
\newcommand{\Hz}{\ \tx{Hz}}
\newcommand{\sinc}{\tx{sinc}}
\newcommand{\tr}{\mathrm{tr}}
\newcommand{\diag}{\mathrm{diag}}
\newcommand{\MAI}{\tx{MAI}}
\newcommand{\ISI}{\tx{ISI}}
\newcommand{\IBI}{\tx{IBI}}
\newcommand{\CN}{\tx{CN}}
\newcommand{\CP}{\tx{CP}}
\newcommand{\ZP}{\tx{ZP}}
\newcommand{\ZF}{\tx{ZF}}
\newcommand{\SP}{\tx{SP}}
\newcommand{\MMSE}{\tx{MMSE}}
\newcommand{\MINF}{\tx{MINF}}
\newcommand{\RC}{\tx{MP}}
\newcommand{\MBER}{\tx{MBER}}
\newcommand{\MSNR}{\tx{MSNR}}
\newcommand{\MCAP}{\tx{MCAP}}
\newcommand{\vol}{\tx{vol}}
\newcommand{\ah}{\hat{g}}
\newcommand{\tg}{\tilde{g}}
\newcommand{\teta}{\tilde{\eta}}
\newcommand{\heta}{\hat{\eta}}
\newcommand{\uh}{\m{\hat{s}}}
\newcommand{\eh}{\m{\hat{\eta}}}
\newcommand{\hv}{\m{h}}
\newcommand{\hh}{\m{\hat{h}}}
\newcommand{\Po}{P_{\mathrm{out}}}
\newcommand{\Poh}{\hat{P}_{\mathrm{out}}}
\newcommand{\Ph}{\hat{\gamma}}
\newcommand{\mat}[1]{\begin{matrix}#1\end{matrix}}
\newcommand{\ud}{^{\dagger}}
\newcommand{\C}{\mathcal{C}}
\newcommand{\nn}{\nonumber}
\newcommand{\nInf}{U\rightarrow \infty}

%% file: problem_formulation_version2.bbl
\vspace{-0.1cm}
\begin{thebibliography}{1}
\bibitem{P_Zhang}
D. G{\"u}nd{\"u}z, Z. Qin, I. E. Aguerri, H. S. Dhillon, Z. Yang, A. Yener, K. K. Wong, and C.-B. Chae, ``Beyond transmitting bits: Context, semantics, and task-oriented communications,'' \emph{IEEE J. Sel. Areas Commun.}, vol. 41, no. 1, pp. 5-41, Jan. 2023.

\bibitem{W_Tong}
Q. Lan, D. Wen, Z. Zhang, Q. Zeng, X. Chen, P. Popovski, and K. Huang, ``What is semantic communication? A view on conveying meaning in the era of machine intelligence,'' \emph{J. Commun. Inf. Networks}, vol. 6, no. 4, pp. 336-371, Dec. 2021.

\bibitem{Z_Yang}
Z. Yang, M. Chen, G. Li, Y. Yang, and Z. Zhang, ``Secure semantic communications: Fundamentals and challenges,'' \emph{arXiv preprint arXiv:2301.01421}, 2023.

\bibitem{8}
J. Kang, H. Du, Z. Li, Z. Xiong, S. Ma, D. Niyato, and Y. Li, ``Personalized saliency in task-oriented semantic communications: Image transmission and performance analysis,'' \emph{IEEE J. Sel. Areas Commun.}, vol. 41, no. 1, pp. 186-201, Jan. 2023.

\bibitem{Huiqiang_Xie}
H. Xie, Z. Qin, G. Y. Li, and B.-H. Juang, ``Deep learning enabled semantic communication systems,'' \emph{IEEE Trans. Signal Process.}, vol. 69, pp. 2663-2675, Apr. 2021.

\bibitem{z_Weng}
Z. Weng, Z. Qin, X. Tao, C. Pan, G. Liu, and G. Y. Li, ``Deep learning enabled semantic communications with speech recognition and synthesis,'' \emph{IEEE Trans. Wireless Commun.}, early access. doi: 10.1109/TWC.2023.3240969.

\bibitem{9}
G. Zhang, Q. Hu, Z. Qin, Y. Cai, G. Yu, X. Tao, and G. Y. Li, ``A unified multi-task semantic communication system for multimodal data,'' \emph{arXiv preprint arXiv: 2209.07689}, 2022.

\bibitem{Lei_Yan}
L. Yan, Z. Qin, R. Zhang, Y. Li, and G. Y. Li, ``Resource allocation for text semantic communications,'' \emph{IEEE Wireless Commun. Lett.}, vol. 11, no. 7, pp. 1394-1398, Jul. 2022.

\bibitem{Xidong_Mu}
X. Mu, Y. Liu, L. Guo, and N. Al-Dhahir, ``Heterogeneous semantic and bit communications: A semi-NOMA scheme,'' \emph{IEEE J. Sel. Areas Commun.}, vol. 41, no. 1, pp. 155-169, Jan. 2023.

\bibitem{10}
L. Yan, Z. Qin, R. Zhang, Y. Li, and G. Y. Li, ``QoE-Aware resource allocation for semantic communication networks'', in \emph{Proc. IEEE Global Commun. Conf. (GLOBECOM)}, Rio de Janeiro, Brazil, Dec. 2022.

\bibitem{Q_shi}
Q. Shi and M. Hong, ``Penalty dual decomposition method for nonsmooth nonconvex optimization Part I: Algorithms and convergence analysis,'' \emph{IEEE Trans. Signal Process.}, vol. 68, pp. 4108-4122, Jun. 2020.

\bibitem{CVX}
M. Grant and S. Boyd, CVX: Matlab software for disciplined convex programming, version 2.1. [Online]. Available: http://cvxr.com/cvx.

\bibitem{S_Boyd}
S. Boyd and L. Vandenberghe, Convex Optimization. Cambridge, U.K.: Cambridge Univ. Press, 2004.
\end{thebibliography}
